\documentclass[journal]{IEEEtran}
\usepackage{amsmath,amssymb}
\usepackage{bm}
\usepackage{makecell}
\usepackage{cite}
\usepackage{graphicx}
\usepackage{xcolor}
\usepackage{algorithm}
\usepackage{algorithmic}

\usepackage{color}
\usepackage{url}
\usepackage[hidelinks]{hyperref}

\newtheorem{theorem}{\textbf{Theorem}}

\newtheorem{assumption}{\textbf{Assumption}}

\newcommand{\R}{\mathbb{R}}
\newcommand{\Aset}{\mathcal{A}}
\newcommand{\Nset}{\mathcal{N}}
\newcommand{\Mset}{\mathcal{M}}
\newcommand{\Lset}{\mathcal{L}}

\newcommand{\Prop}{\mathcal{P}}
\newcommand{\Cphys}{\mathcal{C}_{\mathrm{phys}}}
\newcommand{\Coper}{\mathcal{C}_{\mathrm{oper}}}
\newcommand{\Hset}{\mathcal{H}}
\newcommand{\Jphys}{\mathcal{J}_{\mathrm{phys}}}
\newcommand{\Joper}{\mathcal{J}_{\mathrm{oper}}}
\newcommand{\Rset}{\mathcal{R}}
\newcommand{\Zcert}{\mathcal{Z}}

\begin{document}

\title{xTRUCE: A Provably Safe Arbiter for Multi-xApp Conflict Mitigation in Agentic O-RAN}
\author{
	Le Xia,~\IEEEmembership{Member,~IEEE},
	Rose Qingyang Hu,~\IEEEmembership{Fellow,~IEEE},
    Paul S. Kudyba,~\IEEEmembership{Student Member,~IEEE},\\
    Zhenlin An,~\IEEEmembership{Member,~IEEE},
	and Haijian Sun,~\IEEEmembership{Senior~Member,~IEEE}
	\thanks{
	Le Xia and Rose Qingyang Hu are with the Bradley Department of Electrical and Computer Engineering, Virginia Tech, Blacksburg, VA 24061, USA (e-mail: \{lexia, rosehu\}@vt.edu).
	
	Paul Kudyba, Zhenlin An, and Haijian Sun are with the School of Electrical and Computer Engineering, University of Georgia, Athens, GA 30602, USA (e-mail: \{paul.kudyba, zhenlin.an, hsun\}@uga.edu).
}
\thanks{The source code of xTRUCE will be released on GitHub upon publication.}
}	
\maketitle

\begin{abstract}
The open radio access network (O-RAN) is evolving toward agentic operation, where large language model (LLM)-driven xApps/rApps generate control proposals under operator intents.
However, such proposals may be conflicting, infeasible, or hallucinated, and no existing system jointly provides proposal-independent safety, priority-aware reconciliation, and traceable feedback.
To this end, we propose a provably safe arbiter, namely xTRUCE, in the near-real-time (Near-RT) RAN intelligent controller for mitigating multi-xApp conflicts in gNB control.
We first develop a structured xApp proposal interface and a three-layer constraint hierarchy that places physical limits and operator-defined rules above relaxable performance targets, alongside a dual-timescale control action space.
A two-stage arbitration mechanism then minimizes target shortfalls in the operator-priority order to finalize safe E2 actions within the Near-RT latency budget, while returning conflict certificates to xApps and the operator for renegotiation.
Finally, we implement xTRUCE in a multi-cell O-RAN case study, and evaluate its multi-process prototype through simulations with live API-backed LLM xApps and over-the-air experiments on OpenAirInterface/FlexRIC-based O-RAN stacks.
Results show that xTRUCE ensures gNB control safety with $100\%$ protected services despite severe proposal hallucinations, achieves priority-consistent performance satisfaction under overload, efficiently guides LLM intent renegotiation via certificates, and keeps a delay-safe E2 control loop.

\end{abstract}


\IEEEpeerreviewmaketitle
\bstctlcite{IEEEtranBSTcontrol}

\section{Introduction}\label{sec:intro}
\IEEEPARstart{T}{he} open radio access network (O-RAN) architecture is reshaping cellular networks on the road toward 6G, as it replaces closed and monolithic base stations with multi-vendor disaggregated units, open interfaces, and hierarchical RAN intelligent controllers (RICs)~\cite{polese2023understanding,polese2024empowering}.
Upon this architecture, network intelligence is implemented as pluggable control applications, among which the xApps hosted by the near-real-time (Near-RT, $10$ ms--$1$ s) RIC close measurement, decision, and action loops directly against the live RAN~\cite{bonati2021intelligence}.
Typically, each xApp operates one management objective, e.g., quality of service (QoS) assurance or energy saving, and such data-driven xApps have already been developed and validated at scale on open experimental platforms~\cite{polese2023coloran,tsampazi2025pandora}.
Notably, these xApps may be developed and supplied by different vendors, but they can observe and steer the same cells, users, and spectrum through the shared E2 interface of the Near-RT RIC layer.

Meanwhile, the rapid advancement of large language models (LLMs) is bringing the 6G vision of AI-native networks within reach, and recent studies anticipate that LLMs will undertake management duties ranging from translating operator intents to directly making gNB control decisions~\cite{letaief2019roadmap,zhou2025llmtelecom,bariah2024genai}.
Specifically, intent-based networking translates high-level operator goals into machine-actionable policies~\cite{leivadeas2023intent}, while agentic systems enable LLMs with autonomous reasoning capabilities to interpret these policies and generate the resulting management proposals across O-RAN functions and timescales~\cite{llmmultiagent2024,shao2024wirelessllm,maxenti2026autoran,agentran2026}.
Merging the two trends, the O-RAN ecosystem is evolving toward the agentic O-RAN, where LLM-driven non-real-time (Non-RT, $\geqslant 1$ s) rApp agents interpret operator intents and multiple LLM-driven Near-RT xApp agents generate proposals of gNB control decisions collaboratively and autonomously.
Nevertheless, the multi-xApp proposals may turn out mutually conflicting, physically infeasible, malformed, or simply hallucinated at run time, and the same concern has already arisen in scripted, learned, and third-party xApps~\cite{hoffmann2023open}.

In fact, inter-xApp conflicts have been a long-standing concern in the O-RAN even before the agentic evolution~\cite{del2025pacifista}, since a closed single-vendor stack reconciles its control logic at design time, whereas the xApps of O-RAN cannot be co-engineered due to their multi-vendor origins.
Being aware of this, the O-RAN specifications dedicatedly reserve a \emph{conflict mitigation} function inside the Near-RT RIC platform and outline its architectural considerations~\cite{oranwg3ricarch,oranConflictMitigation2024}.
However, they do not prescribe a per-decision method that jointly admits conflicting targets, relaxes the infeasible ones according to operator priorities, and certifies the resulting action.
Therefore, governing such untrusted intelligent agents at run time, instead of blindly trusting them, becomes the prerequisite for the agentic evolution to be deployable in the real-world O-RAN.

Lately, conflict handling in the O-RAN has attracted several noteworthy research efforts.
Adamczyk~\emph{et al.}~\cite{adamczyk2023conflict} categorized inter-xApp conflicts into direct, indirect, and implicit types, and devised a detection-centric mitigation framework for the Near-RT RIC.
Similarly, del Prever~\emph{et al.}~\cite{del2025pacifista} developed PACIFISTA, a framework that profiles candidate xApps on an experimental platform and statistically evaluates their conflicts before deployment, and their latest follow-up further predicted the key performance indicators (KPIs) impact of such conflicts from the same profiles~\cite{brachdelprever2026predicting}.
For run-time reconciliation, Wadud~\emph{et al.} developed QACM in~\cite{wadud2024qacm} to select the values of conflicting parameters according to the QoS requirements of the involved xApps, and further mitigated a concrete conflict between energy saving and mobility management in the same spirit in~\cite{conflictenergy2025}.
From the orchestration perspective, D'Oro~\emph{et al.}~\cite{doro2024orchestran} proposed OrchestRAN to decide which intelligence runs where, preventing two models from steering the same parameter by construction.
All these efforts confirm both the severity and the industrial relevance of the problem.

In parallel, there have been some related works addressing the safety control issue of data-driven, standalone xApps in O-RAN.
Nagib~\emph{et al.}~\cite{nagib2024safedrl} stabilized a deep reinforcement learning (DRL)-based slicing xApp against unseen network conditions through hybrid transfer learning, and their follow-up SafeSlice~\cite{nagib2025safeslice} further added a safety layer that projects each resource-allocation action onto the nearest latency-safe one.
Hossen~\emph{et al.}~\cite{opentwin2026} replayed the closed control loop on a digital twin, such that an xApp is exercised against realistic KPIs before it touches the live RAN.
In addition, Fiandrino~\emph{et al.}~\cite{fiandrino2023explora} attached graph-based explanations to the decisions of DRL xApps, letting the operator inspect and steer a running agent.
However, these safeguards only target single-xApp O-RAN scenarios and are not capable of arbitrating conflicting proposals from multiple xApps at runtime, let alone handling LLM-driven agents that might generate hallucinated proposals.
To the best of our knowledge, no existing work delivers a unified solution that turns conflicting and infeasible multi-xApp proposals into a verified safe gNB action while balancing operator priorities and Near-RT latency budget.

In summary, we are encountering the following three fundamental challenges when pursuing such a realization:
\begin{itemize}
\item \emph{Challenge 1: How to guarantee gNB control safety under arbitrary, untrusted, and possibly hallucinated proposals from LLM-driven xApp agents?}
The content of multiple xApp proposals may be well-formed and authorized yet remain infeasible or mutually conflicting~\cite{wadud2024qacm}.
Hence, safety must be enforced in the Near-RT RIC platform and be entirely independent of the proposal contents, while ensuring that every executed gNB action satisfies the physical limits and operator-defined rules simultaneously.
\item \emph{Challenge 2: How to reconcile jointly infeasible KPI targets in alignment with operator priorities and provide actionable feedback for proposal renegotiation?}
When multi-xApp proposals conflict, their KPI targets may not be jointly feasible under the current RAN capacity and operating constraints, and some targets must therefore be relaxed according to operator-defined priorities.
Especially for some hard targets, their resulting shortfalls and limiting constraints are essential to be reported in a machine-readable form to guide proposal renegotiation. 
\item \emph{Challenge 3: How to issue a safe gNB action within the Near-RT latency budget over a mixed discrete and continuous action space?}
Configuration gNB control decisions, such as cell activation and user re-association, evolve on a slower timescale than the regular radio resource decisions, which makes the joint optimization combinatorial.
Moreover, the solver of action arbitration may overrun the required delay demand, and thus a safe action must be decided by the decision deadline.
\end{itemize}

\begin{figure*}[!t]
\centering
\includegraphics[width=0.97\textwidth]{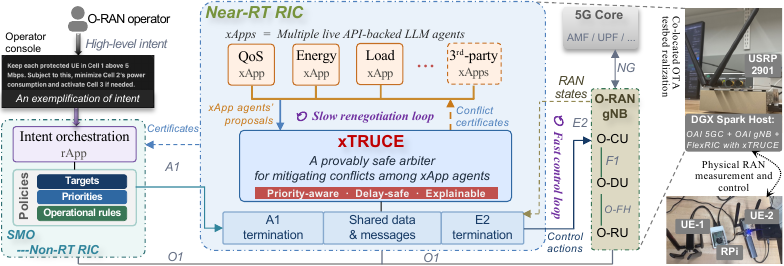}
\caption{An overview of xTRUCE-enabled agentic O-RAN control architecture, alongside its co-located OTA testbed realization.}
\label{fig:arch}
\end{figure*}
To address the above challenges for agentic O-RAN, in this paper, we propose the \underline{x}App \underline{T}arget \underline{R}elaxation \underline{U}nder \underline{C}onflicts with delay-safe \underline{E}xecution, namely \textit{xTRUCE}, in the Near-RT RIC with provable guarantees.
With xTRUCE, xApp agents influence the gNB control exclusively through structured proposals, and the core arbiter computes a safe executed action and a conflict certificate at each decision epoch.
In tests with multiple conflicting xApps, xTRUCE preserves all physical and operator-defined rigid limits, while allocating inevitable KPI-target shortfalls according to operator priority.
It also guides live LLM xApp agents to successfully correct infeasible intents within only $3$ rounds of certificate feedback, leaving targets unmet in only $24$--$29\%$ of rounds versus $79$--$99\%$ for the benchmark modes.
Further on the over-the-air (OTA) O-RAN testbed, a safe action is always output in every $100$-ms runtime epoch, even when arbitration occasionally exceeds the time threshold under severely hallucinated proposals.
In a nutshell, our main contributions are summarized below.
\begin{itemize}
\item We first place xTRUCE in the conflict-mitigation path of the Near-RT RIC, where a structured proposal gateway confines every untrusted xApp to machine-readable proposals rather than direct RAN control.
We further organize physical limits, operator-defined rules, and xApps' KPI targets into a three-layer constraint hierarchy, in which every executed gNB action obeys the two rigid layers and only the KPI-target layer is relaxable.
\item We then develop a two-stage priority-aware arbitration to reconcile conflicting KPI targets.
Specifically, Stage~I sequentially minimizes hard-target relaxations according to operator-defined priorities, while Stage~II finalizes a safe gNB action via the E2 control loop.
Here, we theoretically prove this priority guarantee and show that Stage~I introduces no relaxation whenever all hard targets are jointly feasible by Theorem~\ref{th:priority}.
Afterward, the resulting target shortfalls and limiting-constraint prices are collected in a machine-readable conflict certificate and returned to guide xApps and the operator on proposal renegotiation.
These address Challenges~1 and~2.
\item We identify a dual-timescale action rule and devise a delay-safe action execution mechanism that always retains a verified action to support slow configuration updates and fast resource control under the Near-RT delay limit.
As such, even when the two-stage arbitration does not finish in time, a safe action remains available for RAN control.
We further instantiate a multi-cell O-RAN control case study, in which xTRUCE is theoretically proved to provision safe actions and certificates with exact prices by Theorem \ref{th:convex}.
These tackle Challenge~3.
\item To validate the deployability of xTRUCE, we evaluate its multi-process prototype in Python-based simulations with independently attached live ChatGPT-5.6 and Claude-Sonnet-5 agents, and on an OTA O-RAN-compliant testbed supported by OpenAirInterface (OAI) gNB and Flexible RIC (FlexRIC).
The results verify its proposal-independent gNB control safety, priority-consistent KPI satisfaction, certificate-guided LLM renegotiation, and Near-RT-delay-safe execution over the E2 interface loop.
\end{itemize}

The remainder of this paper is organized as follows.
Section~\ref{sec:arch} first introduces xTRUCE to the O-RAN alongside its design requirements.
Section~\ref{sec:gov} then establishes its system model and mathematical guarantees.
In Section~\ref{sec:usecase}, xTRUCE is instantiated for multi-cell O-RAN control.
Finally, simulation and OTA experimental results are demonstrated and analyzed in Section~\ref{sec:eval}, followed by the conclusions in Section~\ref{sec:concl}.

\section{xTRUCE for Agentic O-RAN}\label{sec:arch}

\subsection{Agentic O-RAN Control Path}
\label{sec:arch-oran}
Fig.~\ref{fig:arch} elaborates the deployment of xTRUCE within an end-to-end control path of agentic O-RAN.
It begins with the O-RAN operator expressing a high-level intent in natural language through the operator console, which is processed by the service management and orchestration (SMO) framework.
The SMO hosts the Non-RT RIC, which has an intent-orchestration rApp to precisely translate this intent into machine-actionable policy guidance that specifies the KPI targets, operator priorities, and operational rules.
Then, data regarding these policies are shared with multiple xApps in the Near-RT RIC via the A1 termination~\cite{etsiA1Gap2025}.
To be agentic and intelligent, the rApp and xApps employ the authorized APIs to access live LLM agents (e.g., Claude-Sonnet-5~\cite{anthropic2026sonnet5} and ChatGPT-5.6~\cite{openai2026gpt56terra}), and they also support local scripting- and learning-based agents.
Meanwhile, the Near-RT RIC communicates with the O-RAN central unit (O-CU) and distributed unit (O-DU) via the E2 interface and controls the O-RAN radio unit (O-RU).
Specifically, E2 indication entering the Near-RT RIC carries the measurements of RAN states, while E2 control leaving it delivers the gNB control decisions through applicable E2 service models (E2SMs)~\cite{etsiE2Gap2024}.
Additionally, the O1 interface serves management and persistent configuration, and the NG interface connects the gNB to the 5G Core~\cite{polese2023understanding}.

Besides hosting the xApps, the Near-RT RIC platform also provides a series of shared functions, among which the O-RAN specifications explicitly identify \emph{multi-xApp conflict mitigation} as a core platform capability and outline its architectural considerations~\cite{oranwg3ricarch,oranConflictMitigation2024}.
However, its standards still leave open which limits are never relaxed, how infeasible targets degrade along the operator priorities, what gNB action goes out if the arbitration hits the Near-RT latency budget, or what should be reported back to the xApps.
In view of this, xTRUCE fills exactly this gap as an indispensable conflict-arbitration stage standing before the E2 control decision is made.
Note that xTRUCE's scope only covers the E2 control-related conflicts among agentic xApps based on the same operator intent, whereas the issues outside this, such as xApp authentication or multi-rApp conflicts, are not considered.

\subsection{xTRUCE Workflow and Design Goals}
\label{sec:arch-loops}
Before xTRUCE works, each xApp first performs inference on the Non-RT policies received via A1 and the RAN states received via E2 in order to generate a proposal for a certain gNB management objective, e.g., QoS assurance, energy saving, or load balancing~\cite{opentwin2026}.
Herein, no xApp directly makes any control decisions, and instead, each xApp submits a machine-readable proposal through a \emph{proposal gateway}, which is essentially an intra-RIC interface distinct from both A1 and the E2SMs.
In detail, the proposal gateway checks the identity, authorized scope, schema, and validity period of the proposals of each xApp, and any malformed, expired, or out-of-scope requests will be discarded before arbitration.
It is also worth pointing out that whatever proposal passes these checks still remains untrusted and must be subsequently arbitrated through xTRUCE, since it may be individually infeasible, mutually conflicting, or simply wrong due to agentic hallucination.

\begin{figure*}[!t]
\centering
\includegraphics[width=0.85\textwidth]{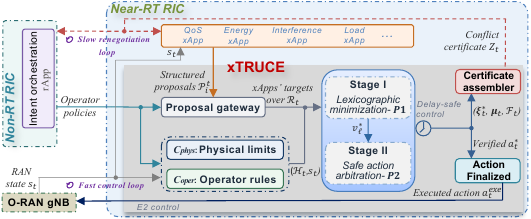}
\caption{The proposed xTRUCE for safe multi-xApp arbitration, delay-safe RAN control, and certificate-guided conflict renegotiation in the Near-RT RIC.}
\label{fig:xTRUCE}
\end{figure*}
Within each arbitration epoch of xTRUCE, the latest valid xApp proposals are first collected, and one best safe action is then computed and issued via E2.
In parallel, xTRUCE returns a structured \emph{conflict certificate} to xApps and the operator, stating which KPI targets were met, which were relaxed and by how much, and which limits were hindering further improvement.
As such, two closed loops emerge in the xTRUCE workflow, as marked in Fig.~\ref{fig:arch}.
The fast control loop starts with receiving RAN states and ends with emitting gNB control actions in each arbitration epoch.
The slow renegotiation loop is an asynchronous feedback path, where the certificates let the xApp agents revise unrealistic targets and give the operator a concise account of persistent conflicts.
For delay-safe consideration, an accepted xApp proposal stays valid over a declared time window, and xTRUCE always decides based on the latest valid proposals and never waits for an LLM response on delay-sensitive services.

Accordingly, the design and development of xTRUCE in agentic O-RAN aims at four specific goals as follows:
\begin{itemize}
\item \textbf{Proposal-independent gNB control safety:} Every gNB control action output from the arbiter must not violate the modeled physical limits and operator-defined rules, independently of the contents of all admitted xApp proposals.
\item \textbf{Priority-consistent KPI satisfaction:} When KPI targets cannot all be met under overload, their shortfalls are minimized based on operator-assigned priorities to satisfy high-priority performance targets as much as possible.
\item \textbf{Multi-xApp conflict certificate feedback:} At each arbitration epoch, xTRUCE must return a machine-readable conflict certificate for renegotiation.
The certificate reports the target relaxations and limiting constraints when a new arbitration result is available, or indicates that no new optimality information is available after a timeout.
\item \textbf{Near-RT RIC delay-safe action execution:} A verified control action must be available and issued within each epoch, even when the complete arbitration does not finish within the Near-RT latency budget (less than $1$s).
\end{itemize}

\section{Proposed xTRUCE System}\label{sec:gov}
This section details the xTRUCE system design in Fig.~\ref{fig:xTRUCE}, covering its control abstraction, constraint hierarchy, two-stage arbitration, certificate feedback, and delay-safe guarantee.

\subsection{Control Actions and Structured xApp Proposals}
\label{sec:gov-abs}
Time is equally divided into discrete, consecutive epochs $t=0,1,2,\cdots$, and at each epoch $t$, the physical RAN state\footnote{The state encompasses channel gains, queue backlogs, and traffic arrivals, etc., which are objective conditions that can be measured in near real time.}, denoted as $s_t$, is known to the Near-RT RIC layer via the E2 interface.
Let $\Nset=\{1,\cdots,N\}$ be the set of xApp agents deployed, where each agent $i \in \Nset$ operates towards a specific KPI and generates a corresponding proposal when needed.
After processing the valid proposals, xTRUCE issues a gNB control action, denoted as $a_t\in\Aset$, where $\Aset$ is an action space containing finite numbers of control variables, covering transmit power/spectrum allocation, cell on/off states, and cell-to-user steering.
To accommodate different control timescales, we split $a_t$ into $a_t=\big(a_t^{\mathrm{cfg}},\,a_t^{\mathrm{rt}}\big)$.
Here, $a_t^{\mathrm{cfg}}$ is the set of slowly reconfigured variables that are updated every $T^{\mathrm{cfg}}$ ($T^{\mathrm{cfg}} > 1$) epochs and stay constant in between, while $a_t^{\mathrm{rt}}$ stands for the set of variables that should be redetermined at every epoch $t$.

In addition, let $\Mset=\{1,\cdots,M\}$ index the KPIs exposed by the O-RAN, and their realized values are given by
\begin{equation}
\phi(a_t,s_t)=\big(\phi_1(a_t,s_t),\cdots,\phi_M(a_t,s_t)\big)^{T},
\label{eq:kpimap}
\end{equation}
where $\phi_m(a_t,s_t)$ is the performance value regarding KPI $m$ resulting from executing action $a_t$ under RAN state $s_t$.

Consider that at epoch $t$, each xApp agent $i$ submits its requests only through a \emph{structured proposal} in the form of
\begin{equation}
\Prop_i^{t} =\big\{\big(\theta_{i,m},\,\mathrm{type}_{i,m}\big)\big\}_{m\in\Mset_i}.
\label{eq:proposal}
\end{equation}
Specifically, $\Mset_i\subseteq\Mset$ confines the scope of KPIs for which agent $i$ is responsible, $\theta_{i,m}\in\R$ is its requested performance target of KPI $m$, and $\mathrm{type}_{i,m}\in\{\mathrm{hard},\mathrm{soft}\}$ identifies a hard target or a soft preference.
For illustration, let $\Rset_t$ collect the set of $(i,m)$ pairs covered by valid proposals at epoch $t$.
Notably, each proposal expires unless renewed, and any malformed, expired, or out-of-scope proposals will be discarded by the proposal gateway,\footnote{Agent identity and authorization metadata can be attached to \eqref{eq:proposal} and checked by the same gateway~\cite{etsiORANSecurity2025}. Relevant agentic security research is beyond the scope of this paper and will not be discussed in depth.} as shown in Fig.~\ref{fig:xTRUCE}.
We further assume that the operator assigns each hard KPI target (i.e., $\theta_{i,m}$ with $\mathrm{type}_{i,m}=\mathrm{hard}$) a fixed priority class $\ell_{i,m}\in\Lset=\{1,\cdots,L\}$, where a smaller $\ell_{i,m}$ denotes a higher priority for target satisfaction and no xApp can change the ordering.

Given $\Prop_i^{t}$ at epoch $t$, xTRUCE outputs action $a_t$ in the current RAN state $s_t$, thus yielding a performance shortfall function $g_{i,m}(a_t,s_t)$ to measure the gap between the expected and the actual performance of KPI $m$. That is, $\forall (i,m)\in \Rset_t$,
\begin{equation}
g_{i,m}(a_t,s_t)=d_m\big(\theta_{i,m}-\phi_m(a_t,s_t)\big).
\label{eq:req}
\end{equation}
Here, $d_m\in\{1,-1\}$ is the preference direction, where $d_m=1$ means that a larger $\phi_m(a_t,s_t)$ is better (e.g., rate) and $d_m=-1$ the opposite (e.g., energy).
Obviously, $g_{i,m}(a_t,s_t)\leqslant 0$ holds only when the KPI target is met, while a positive $g_{i,m}(a_t,s_t)$ equals the shortfall by which the target is missed.

Different from existing agentic O-RAN designs, which implicitly trust the agentic reasoning of xApps, our framework is built on the opposite premise by the following assumption:
\begin{assumption}[Untrusted agentic proposals]\label{as:untrusted}
No condition is imposed on xApp proposals admitted by the gateway.
They may be individually unachievable, mutually contradictory, or wrong, and xTRUCE never relies on their quality.
\end{assumption}

\subsection{Three-Layer Constraint Hierarchy}\label{sec:gov-layers}
The action $a_t$ output from xTRUCE at epoch $t$ must obey three layers of constraints.
The first layer contains physical RAN constraints, such as a transmit power budget, that cannot be relaxed by any xApp or operator intent.
In this case, given any network state $s_t$, we define the set of admitted actions as
\begin{equation}
\Cphys(s_t)=\big\{a_t\in\Aset \mid c_j(a_t,s_t)\leqslant 0,\ \forall j\in\Jphys\big\},
\label{eq:cphys}
\end{equation}
where $\Jphys$ indexes all preset physical constraints, and $c_j(a_t,s_t)$ is the formulation of $j$-th constraint rewritten by moving all its terms to the left-hand side, thus a positive $c_j(a_t,s_t)$ is exactly the shortfall at limit $j$.
The second-layer constraint is derived from rigid operational rules of O-RAN to maintain network availability, such as a guaranteed rate for protected users, which can only be modified by the operator and may involve the action executed at the previous epoch, denoted as $a_{t-1}^{\mathrm{exec}}$.
Likewise, its admitted action set is
\begin{equation}
\Coper(s_t,a_{t-1}^{\mathrm{exec}})\!=\!\{a_t\!\in\!\Aset \mid e_j(a_t,s_t,a_{t-1}^{\mathrm{exec}})\leqslant 0,\forall j\in\Joper\},
\label{eq:coper}
\end{equation}
where $\Joper$ indexes all preset operational constraints, and $e_j(a_t,s_t,a_{t-1}^{\mathrm{exec}})$ indicates the shortfall of its $j$-th constraint.

Together, the two rigid layers confine a safe action set
\begin{equation}
\Hset_t=\Cphys(s_t)\cap\Coper(s_t,a_{t-1}^{\mathrm{exec}}),
\label{eq:safeset}
\end{equation}
such that each action $a_t \in \Hset_t$ satisfies the physical and operator limits independently of the xApp proposals. Here, $\Hset_t$ is assumed to be compact.
In addition, we recognize the third-layer constraints by incorporating the KPI targets of xApps.
Under Assumption~\ref{as:untrusted}, their target inequalities need not be jointly feasible and are therefore not included in $\Hset_t$, but will be taken into account during subsequent processing.
Clearly, the first two layers decide which actions are admissible, while the third determines which admissible action is preferred.

To ensure that a valid E2 control action exists at every epoch, we define a baseline action as $a_t^{\mathrm{base}}=\pi_{\mathrm{base}}(s_t,a_{t-1}^{\mathrm{exec}})$, which retains the configuration component of $a_{t-1}^{\mathrm{exec}}$ and computes a conservative real-time component $a_t^{\mathrm{rt}}$ satisfying both rigid layers in observation of $s_t$.
In this way, $a_t^{\mathrm{base}}$ is easy to obtain and serves as the floor action whenever xTRUCE cannot output a better action in time.
Furthermore, we make the following assumption:
\begin{assumption}[Certified baseline action]\label{as:baseline}
At each epoch $t$, let $a_t^{\mathrm{base}}\in\Hset_t$, which also implies that $\Hset_t\neq\varnothing$.
\end{assumption}

By Assumption~\ref{as:baseline}, the joint action feasibility between the physical limits and the operator rules is guaranteed.

\subsection{Two-Stage Priority-Aware Arbitration}\label{sec:gov-arb}
When conflicts happen, the hard targets of multiple xApps may no longer be jointly met within $\Hset_t$, hence xTRUCE must decide how much intent to give up following the operator's preset priorities.
To this end, we devise a two-stage arbitration mechanism, where Stage I determines the extent to which the hard targets need to be relaxed from highest to lowest priority (i.e., $\ell_{i,m}$), and Stage II finalizes which safe action to execute.

\subsubsection{Stage I}
For each hard-target pair $(i,m)\in\Rset_t$, we first introduce a non-negative slack $\xi_{i,m}$ through $g_{i,m}(a_t,s_t)\leqslant\xi_{i,m}$, thus $\xi_{i,m}$ exactly equals the shortfall forced to accept.
With this, the overall shortfall across all hard-target pairs at the $\ell$-th priority class (i.e., $\ell_{i,m}=\ell$) is now measured by
\begin{equation}
V_\ell(\bm{\xi})=\sum_{(i,m)\in\Rset_t:\,\ell_{i,m}=\ell,\ \mathrm{type}_{i,m}=\mathrm{hard}}
\beta_{i,m}\,\xi_{i,m}^{2},
\label{eq:classviol}
\end{equation}
where $\beta_{i,m}>0$ is the operator-assigned target weight, stating the relative importance of different targets within the same priority class.
Clearly, $V_\ell(\bm{\xi})=0$ holds only when each hard target in class $\ell$ is met, while the design of squared slacks can discourage concentrating unavoidable shortfall on one target.

At each epoch $t$, Stage~I then solves a lexicographic violation optimization problem as follows:
\begin{align}
\mathbf{P1}:\ \operatorname*{lex\,min}_{a_t\in\Hset_t,\ \bm{\xi}\geqslant 0}\quad
&\big(V_1(\bm{\xi}),\,V_2(\bm{\xi}),\,\cdots,\,V_L(\bm{\xi})\big)\label{P1}\\
{\rm s.t.}\quad
&g_{i,m}(a_t,s_t)\leqslant\xi_{i,m},\ \forall\,\mathrm{hard}\ (i,m)\in\Rset_t.\tag{\ref{P1}a}
\end{align}
Lexicographic minimization stipulates that $V_1(\bm{\xi})$ is minimized first, then holding that value while proceeding to solve $V_2(\bm{\xi})$, and so on, until $V_L(\bm{\xi})$~\cite{ehrgott2005multicriteria}.
When minimizing $V_\ell(\bm{\xi})$, each higher-priority $j<\ell$ is constrained by $V_j(\bm{\xi})\leqslant v_j^{\star}$, and the attained optimum is recorded as $v_\ell^{\star}$.
Note that $\Hset_t$ is never relaxed at any point in the above process, and only xApps' KPI targets are relaxed via $\bm{\xi}$.
Moreover, solving $\mathbf{P1}$ does not ensure a unique optimal action, as many $a_t\in\Hset_t$ may reach the optimum at the same time, e.g., in a case where all hard targets jointly fit, the corresponding actions yield $V_\ell(\bm{\xi})=0$.

\subsubsection{Stage II}
In this stage, xTRUCE is to select a best control action from the action solutions to $\mathbf{P1}$ by minimizing the shortfalls of all soft targets from xApp proposals as well as the inter-epoch action change.
Its optimization problem becomes
\begin{align}
\mathbf{P2}:\ \min_{a_t\in\Hset_t,\ \bm{\xi}\geqslant 0}\
&\sum_{(i,m)\in\Rset_t:\,\mathrm{type}_{i,m}=\mathrm{soft}}\beta_{i,m}
\big[g_{i,m}(a_t,s_t)\big]_{+}^{2}\notag\\
&\qquad\qquad+\eta\,D\big(a_t,a_{t-1}^{\mathrm{exec}}\big)\label{P2}\\
{\rm s.t.}\quad
&g_{i,m}(a_t,s_t)\leqslant\xi_{i,m},\ \forall\,\mathrm{hard}\ (i,m)\in\Rset_t,
\tag{\ref{P2}a}\\
&V_\ell(\bm{\xi})\leqslant v_\ell^{\star}+\epsilon_\ell,\ \forall
\ell\in\Lset.\tag{\ref{P2}b}
\end{align}
Here, $[\,\cdot\,]_{+}=\max\{\cdot\,,0\}$ extracts the shortfall of each soft-target pair $(i,m)\in\Rset_t$, $D(a_t,a_{t-1}^{\mathrm{exec}})\geqslant0$ measures the action change,\footnote{The explicit form of $D(a_t,a_{t-1}^{\mathrm{exec}})$ may vary in different use cases, e.g., it is instantiated as the squared Euclidean distance in the following Section~\ref{sec:usecase}.} and $\eta\geqslant0$ is an operator-assigned weight to control the magnitude of action change.
Constraints (\ref{P2}a) and (\ref{P2}b) require each priority's shortfall to stay within $v_\ell^{\star}$ up to a non-negative tolerance $\epsilon_\ell$, thereby preserving the Stage-I priority ordering.
Under Assumption~\ref{as:baseline} and continuity of the functions in $\mathbf{P1}$ and $\mathbf{P2}$ on the compact set $\Hset_t$, both stages can attain their optima, and the Stage-II action is denoted by $a_t^{\star}$.

Moreover, xTRUCE keeps the configuration $a_t^{\mathrm{cfg}}$ fixed between configuration epochs, while at every $T^{\mathrm{cfg}}$-th epoch, it evaluates the operator-prescribed configuration candidates through $\mathbf{P1}$ and selects the one with the lexicographically smallest $(v_1^{\star},\cdots,v_L^{\star})$.
Then, $\forall a_t\in\Hset_t$, let $V_\ell(a_t)$ denote \eqref{eq:classviol} with $\xi_{i,m}$ substituted by $\xi_{i,m}=[g_{i,m}(a_t,s_t)]_{+}$.
The resulting priority behavior of the two-stage arbitration is stated below:
\begin{theorem}[Two-stage lexicographic minimization]\label{th:priority}
The Stage-I output $(v_1^{\star},\cdots,v_L^{\star})$ is the unique lexicographically minimal vector over $\Hset_t$, and $\forall a_t\in\Hset_t$, either $V_\ell(a_t)=v_\ell^{\star}$, $\forall \ell$, or $V_\ell(a_t)>v_\ell^{\star}$ at the highest priority class $\ell$ where they differ.
Moreover, for arbitrary $\epsilon_\ell\geqslant0$, the Stage-II output satisfies $V_\ell(a_t^{\star})\leqslant v_\ell^{\star}+\epsilon_\ell$, $\forall\ell$, where equality holds if all $\epsilon_{\ell}=0$.
\end{theorem}
\begin{IEEEproof}
Please see Appendix~A.
\end{IEEEproof}

Theorem~\ref{th:priority} points out that reducing any $V_\ell$ below $v_\ell^{\star}$ necessarily increases the shortfall of a higher-priority class above its optimum, and Stage~II selects the final action within the class bounds of (\ref{P2}b).
Furthermore, if some actions jointly satisfy all hard targets, clearly we have $v_\ell^\star=0$.

\subsection{Conflict Certificates and Delay-Safe Execution}
\label{sec:gov-cert}
As shown in Fig.~\ref{fig:xTRUCE}, at each epoch, xTRUCE ultimately returns a machine-readable conflict certificate for xApp agent renegotiation and selects a verified gNB control action, thereby ensuring that no delayed or unverified arbitration results will be executed.
The two outcomes are elaborated below.

\subsubsection{Conflict Certificate with Prices}
At each epoch, the certificate flows along the slow loop to the operator and the relevant xApps (each reading only the entries for its own targets), capturing two key points: the extent to which each hard target was relaxed and which limits of $\Hset_t$ result in such relaxation.
The former is the realized shortfall at $a_t^{\star}$ of $\mathbf{P2}$, denoted as $\bm{\xi}_t^{\star}=[g_{i,m}(a_t^{\star},s_t)]_{+}$.
For the latter, we associate each constraint of $\mathbf{P2}$ with a Lagrange multiplier.
When $\mathbf{P2}$ is convex and satisfies a constraint qualification such as Slater's condition~\cite{boyd2004convex}, an optimal multiplier measures the local change in the optimum caused by relaxing its constraint, and here we refer to it as a \emph{price}.
Such conditions, combined with complementary slackness, further guarantee that a positive price identifies a binding constraint.
Without these, the returned prices are only heuristic local diagnostics without a global guarantee.
By carefully examining $\mathbf{P2}$, let $\mu_{i,m}\geqslant 0$ denote the price of constraint (\ref{P2}a) and $\mu_j\geqslant 0$, $j\in\Jphys\cup\Joper$, that of the $j$-th limit defining $\Hset_t$, which together form the price set $\bm{\mu}_t$.
With these, the conflict certificate is given as
\begin{equation}
	\Zcert_t=\big(\bm{\xi}_t^{\star},\ \bm{\mu}_t,\ \mathcal{F}_t,\ t\big),
\label{eq:cert}
\end{equation}
where $\mathcal{F}_t$ is produced only at every $T^{\mathrm{cfg}}$-th epoch, collecting $(v_1^{\star},\cdots,v_L^{\star})$ at each operator-prescribed candidate $a_t^{\mathrm{cfg}}$, thereby explaining the choice of the best one.
The epoch stamp $t$ prevents an expired certificate from affecting later decisions.
In short, $\Zcert_t$ provides the information, in machine-readable form, regarding three issues of concern to an xApp or the operator about what was missed, why, and what would help.

The certificate shapes the slow renegotiation loop of Section~\ref{sec:arch-loops} in two modes.
In the autonomous mode, it is returned to the xApp agents to revise their unmet targets, while in the human-in-the-loop mode, it is reported to the operator to adjust priorities or, exceptionally, the operational limits, which is the only way to change $\Coper$.
Notably, whether and how fast the renegotiation converges depends on the agentic capabilities.

\subsubsection{Delay-Safe Action Execution}
Solving $\mathbf{P1}$ and $\mathbf{P2}$ may exceed the fixed per-epoch time threshold, denoted as $\tau$.
To handle this, the arbiter maintains a standby action $a_t^{\mathrm{stb}}$, first initialized with $a_t^{\mathrm{base}}$ of Assumption~\ref{as:baseline} and then updated at each $T^{\mathrm{cfg}}$-th epoch by the verified Stage-I action associated with the best configuration candidate available by $\tau$.
Incorporating all potential cases, the E2 action ultimately output from xTRUCE for execution, denoted as $a_t^{\mathrm{exec}}$, is determined by
\begin{equation}
a_t^{\mathrm{exec}}=
\begin{cases}
a_t^{\star}, & \text{if a Stage-II action is available by $\tau$},\\
a_t^{\mathrm{stb}}, & \text{else, if updated at $T^{\mathrm{cfg}}$-th epoch},\\
a_{t-1}^{\mathrm{exec}}, & \text{else, if } a_{t-1}^{\mathrm{exec}}\in\Hset_t \text{ still holds},\\
a_t^{\mathrm{base}}, & \text{otherwise}.
\end{cases}
\label{eq:anytime}
\end{equation}

Note that every action selected by \eqref{eq:anytime} always belongs to $\Hset_t$, since the first two are verified, $a_{t-1}^{\mathrm{exec}}$ is rechecked, and $a_t^{\mathrm{base}}$ satisfies Assumption~\ref{as:baseline}.
Most importantly, the physical and operational limits are completely decoupled from the quality of the agents, and only the achievability of their intents degrades by priority as proposals worsen or conflicts arise.
If no verified $a_t^{\star}$ is available by $\tau$, $\bm{\xi}_t^{\star}$ and $\bm{\mu}_t$ are set to $\varnothing$ to indicate that no new Stage-II optimality is available, while $\mathcal F_t$ retains any configuration evaluations completed by then.

\begin{algorithm}[t]
\caption{The Proposed xTRUCE for Agentic O-RAN}
\label{alg:arbiter}
\begin{algorithmic}[1]
\REQUIRE \textit{$\Aset$, $\Nset$, $\phi$, $\{d_m\}_{m\in\Mset}$, $\pi_{\mathrm{base}}$, $T$, $\tau$, $T^{\mathrm{cfg}}$, $\Omega^{\mathrm{cfg}}=\{a_{(n)}^{\mathrm{cfg}}\}_{n=1}^{C}$, and the operator policies ($L$, $\{\epsilon_\ell\}$, $\ell_{i,m}$, $\beta_{i,m}$)}
\ENSURE \textit{$\{a_t^{\mathrm{exec}},\Zcert_t\}_{t=1}^{T}$}
\STATE \textit{Set a certified initial action $a_0^{\mathrm{exec}}$}
\FOR{$t\leftarrow1$ \textit{to} $T$ \textit{(i.e., the fast control loop)}}
  \STATE \textit{Obtain $s_t$ and valid proposals, form $\Rset_t$, calculate $g_{i,m}$ by \eqref{eq:req}, and build $\Hset_t$ using $a_{t-1}^{\mathrm{exec}}$ by \eqref{eq:cphys}--\eqref{eq:safeset}}
  \STATE \textit{Start a timer, set $a_t^{\mathrm{stb}}\leftarrow\pi_{\mathrm{base}}(s_t,a_{t-1}^{\mathrm{exec}})$ and $\mathcal F_t\leftarrow\varnothing$}
  \IF{$t\bmod T^{\mathrm{cfg}}=0$}
    \FOR{$n\leftarrow1$ \textit{to} $C$}
      \STATE \textit{Solve $\mathbf{P1}$ under $a_{(n)}^{\mathrm{cfg}}$ and, if completed by $\tau$, add $\big(a_{(n)}^{\mathrm{cfg}},(v_{1,(n)}^\star,\ldots,v_{L,(n)}^\star)\big)$ to $\mathcal F_t$}
    \ENDFOR
    \STATE \textit{If $\mathcal F_t\neq\varnothing$, set $a_t^{\mathrm{stb}}$ to the verified Stage-I action associated with its lexicographically best entry}
    \STATE \textit{If $|\mathcal F_t|=C$, set $\big(a_t^{\mathrm{cfg}},(v_1^\star,\ldots,v_L^\star)\big)$ to the lexicographically best entry of $\mathcal F_t$}
  \ELSE
    \STATE \textit{Retain the previously executed configuration action and solve $\mathbf{P1}$ for $(v_1^\star,\ldots,v_L^\star)$}
  \ENDIF
  \STATE \textit{If Stage I finishes by $\tau$, attempt to solve $\mathbf{P2}$ once and accept $(a_t^\star,\bm\mu_t)$ only if returned and verified by $\tau$}
  \STATE \textit{Evaluate $\bm\xi_t^\star$ at an accepted $a_t^\star$, or set $(\bm\xi_t^\star,\bm\mu_t)\leftarrow(\varnothing,\varnothing)$ otherwise}
  \STATE \textit{Select $a_t^{\mathrm{exec}}$ by \eqref{eq:anytime}, assemble $\Zcert_t$ by \eqref{eq:cert}, and return the relevant entries to each xApp and the full certificate to the operator (i.e., the slow renegotiation loop)}
\ENDFOR
\end{algorithmic}
\end{algorithm}
\subsection{Algorithm and Complexity Analysis}
To better demonstrate the full picture of xTRUCE for Agentic O-RAN, we summarize its relevant technical points and enclose them in Algorithm~\ref{alg:arbiter}, as shown on the next page.

In terms of the computational complexity of Algorithm~\ref{alg:arbiter}, line~3 processes at most $|\Rset_t|\leqslant NM$ targets per epoch in $\mathcal{O}(NM)$ time and assembles $\Hset_t$ in time linear in the number of limits in \eqref{eq:cphys} and \eqref{eq:coper}.
The dominant cost lies in the optimization calls, i.e., the baseline solve in line~4, the $L$ class solves behind each solve of $\mathbf{P1}$ in line~7 or line~12, and the single solve of $\mathbf{P2}$ in line~14, hence requiring $(L+2)$ solves per regular epoch and $(CL+2)$ per configuration epoch, where $C$ is the size of the configuration candidate set $\Omega^{\mathrm{cfg}}$.
Once instantiated into convex programs, every such call is solvable to any fixed accuracy in polynomial time~\cite{boyd2004convex}.
Meanwhile, verifying an action against $\Hset_t$ in lines~9, 14, and~16 takes one linear pass over the limits and targets, keeping the delay-safe rule \eqref{eq:anytime} affordable at any interruption instant.
As such, denoting $Q_{\mathrm{sol}}$ as the cost of one solve, running Algorithm~\ref{alg:arbiter} over all the $T$ epochs has a polynomial-time overall complexity of $\mathcal{O}\big(T\,NM+(L+2)\,T\,Q_{\mathrm{sol}}+(C-1)\,L\,\lfloor T/T^{\mathrm{cfg}}\rfloor\,Q_{\mathrm{sol}}\big)$.

\section{Case Study: Safe Multi-Cell Control in Agentic O-RAN}\label{sec:usecase}
This section instantiates our xTRUCE for multi-cell agentic O-RAN control by specifying its state, dual-timescale actions, KPI map, rigid constraints, and four xApp roles.

\subsection{Network and E2 Control Actions}\label{sec:uc-net}
Consider $B$ cells indexed by $\mathcal B=\{1,\cdots,B\}$ serving $U$ users indexed by $\mathcal U=\{1,\cdots,U\}$ over $K$ RBs indexed by $\mathcal K=\{1,\cdots,K\}$, each with bandwidth $W$ Hz.
Then, we instantiate the configuration and fast components of $a_t$ as
\begin{gather}
a_t^{\mathrm{cfg}}
=\big(\{\alpha_b(t)\}_{b\in\mathcal B},
\{z_u(t)\}_{u\in\mathcal U}\big),
\label{eq:actinst}\\
a_t^{\mathrm{rt}}
=\big(\{x_{u,k}(t)\}_{u\in\mathcal U,k\in\mathcal K},
\{p_{u,k}(t)\}_{u\in\mathcal U,k\in\mathcal K}\big).
\label{eq:actinstrt}
\end{gather}
Here, $\alpha_b(t)\in\{0,1\}$ is the cell activation/sleep state ($1$ means activation), while $z_u(t)\in\mathcal B$ identifies the active cell serving user $u$, with $\alpha_{z_u(t)}(t)=1$ and $\mathcal U_b(t)=\{u\mid z_u(t)=b\}$.
At the fast timescale, $x_{u,k}(t)\in[0,1]$ is the fraction\footnote{Here, we assume that multiple users can reuse one RB at one time, and multiple RB shares used by one user do not need to be adjacent in frequency.} of epoch $t$ for which RB $k$ is allocated to user $u$, while $p_{u,k}(t)\geqslant0$ is its average transmit power.
Notably, the components of \eqref{eq:actinst} and \eqref{eq:actinstrt} are RAN controls managed at the Near-RT RIC side, where the cell activation, cell-to-user steering, and RB shares are already exposed by existing E2SMs~\cite{polese2024empowering,lacava2024steering}.
Naturally, the cell-level transmit power on RB $k$ is  $P_{b,k}(t)=\sum_{u\in\mathcal{U}_b(t)}p_{u,k}(t)$.

\subsection{RAN State and KPI Models}\label{sec:uc-kpi}
In this case study, the RAN state $s_t$ contains channel gains, inter-cell interference, queue backlogs, and traffic arrivals at epoch $t$.
Let $G_{b,u,k}(t)\geqslant0$ denote the channel power gain from cell $b$ to user $u$ on RB $k$, and write $G_{u,k}(t)=G_{z_u(t),u,k}(t)$ for the serving-cell gain.
Then, let $I_{u,k}(t)\geqslant 0$ be the fixed interference power received by user $u$ on RB $k$ (measured at epoch $t-1$ and treated as known at epoch $t$, which is updated between epochs).
Further, $Q_u(t)\geqslant0$ and $\lambda_u(t)\geqslant0$ denote the downlink queue backlog and per-epoch traffic-arrival amount for user $u$, respectively, both in bits.
In the meantime, the KPI map $\phi$ of~\eqref{eq:kpimap} contains five KPI families, including user rate with $d_m=1$ plus queue backlog, per-cell energy consumption, caused interference, and load use with $d_m=-1$.

First, the bit rate at user $u$ in \textit{bits/s} is given by
\begin{equation}
R_u(a_t,s_t)=\sum_{k\in\mathcal{K}}x_{u,k}(t)W
\log_2\!\bigg(1+\frac{G_{u,k}(t)p_{u,k}(t)}
{x_{u,k}(t)(\sigma^2+I_{u,k}(t))}\bigg),
\label{eq:rate}
\end{equation}
where $\sigma^2$ is the noise power over one RB.
Note that $R_u(a_t,s_t)$ is defined as zero when $x_{u,k}(t)=0$, while $p_{u,k}(t)/x_{u,k}(t)$ is the transmit power used when user $u$ is scheduled.

The queue update over an epoch of length $\tau$, in \textit{bits}, is
\begin{equation}
Q_u(t+1)=\big[Q_u(t)-\tau R_u(a_t,s_t)\big]_{+}+\lambda_u(t).
\label{eq:queue}
\end{equation}
This reflects the queuing delay, since a bounded queue at a given arrival rate implies a bounded delay by Little's law~\cite{little2008little}.

The remaining per-cell KPIs of power consumption and interference in \textit{W}, and dimensionless resource load, are
\begin{gather}
E_b(a_t)=P_b^{\mathrm{cir}}\alpha_b(t)
+\sum_{k\in\mathcal K}P_{b,k}(t),
\label{eq:energy}\\
I_b^{\mathrm{out}}(a_t,s_t)
=\sum_{k\in\mathcal K}\sum_{u\notin\mathcal U_b(t)}
G_{b,u,k}(t)P_{b,k}(t),
\label{eq:leak}\\
\rho_b(a_t)=\frac{1}{K}
\sum_{u\in\mathcal U_b(t)}\sum_{k\in\mathcal K}x_{u,k}(t),
\label{eq:load}
\end{gather}
respectively. Among them, $P_b^{\mathrm{cir}}$ is a fixed circuit power of an active cell $b$, and $\rho_b(a_t)$ is its average RB utilization.
Unlike the measured $I_{u,k}(t)$, $I_b^{\mathrm{out}}(a_t,s_t)$ depends on the selected transmit powers and can therefore be targeted by an xApp.

\subsection{Physical and Operator Limits}\label{sec:uc-rigid}
In the context of~\eqref{eq:actinst} and \eqref{eq:actinstrt}, the physical layer of Section~\ref{sec:gov-layers} is instantiated by the following four limits:
\begin{equation}
\left\{
\begin{aligned}
&c_1:\sum_{k\in\mathcal{K}}P_{b,k}(t)-\alpha_b(t)P_b^{\max}\leqslant 0,\\
&c_2: \sum_{u\in\mathcal{U}_b(t)}x_{u,k}(t)-1\leqslant 0,\\
&c_3: -x_{u,k}(t)\leqslant 0, \ x_{u,k}(t)-1\leqslant 0,\\
&c_4: -p_{u,k}(t)\leqslant0, \ p_{u,k}(t)-x_{u,k}(t)P^{\mathrm{rb}}\leqslant 0.
\end{aligned}
\right.
\label{eq:phys}
\end{equation}
Here, $P_b^{\max}$ and $P^{\mathrm{rb}}$ are the per-cell and per-RB power limits, respectively.
Moreover, constraints $c_1$--$c_4$ bound the total cell power, total share of each RB, individual RB shares, and the transmit power assigned with each share, respectively.

Then, for the operational layer, let $N^{\mathrm{act}}(t)$ and $N^{\mathrm{str}}(t)$ count the cells changing activation state and the users changing serving cell between epochs $t-1$ and $t$, respectively.
Likewise, for all $(b,u,k)$ at epoch $t$, the operator rules are instantiated by the following five limits:
\begin{equation}
\left\{
\begin{aligned}
&e_1: R_u^{\min}-R_u(a_t,s_t)\leqslant 0,\ \forall u\in\mathcal{U}^{\mathrm{prot}},\\
&e_2: |p_{u,k}(t)-p_{u,k}(t-1)|-\Delta^{p}\leqslant 0,\\
&e_3: |x_{u,k}(t)-x_{u,k}(t-1)|-\Delta^{x}\leqslant 0,\\
&e_4: N^{\mathrm{act}}(t)-\Delta^{\mathrm{act}}\leqslant 0,\\
&e_5: N^{\mathrm{str}}(t)-\Delta^{\mathrm{str}}\leqslant 0,
\end{aligned}
\right.
\label{eq:ramp}
\end{equation}
where the last two only apply at configuration epochs.
Constraint $e_1$ enforces the committed rate floor $R_u^{\min}$ for each user in the operator-designated protected user group $\mathcal U^{\mathrm{prot}}$, while $e_2$--$e_5$ limit consecutive-epoch changes in power by $\Delta^{p}$, RB shares by $\Delta^{x}$, cell activation state by $\Delta^{\mathrm{act}}$, and cell-to-user steering by $\Delta^{\mathrm{str}}$, respectively.
Together, \eqref{eq:phys} and \eqref{eq:ramp} determine $\Cphys(s_t)$ and $\Coper(s_t,a_{t-1}^{\mathrm{exec}})$, and hence $\Hset_t$ in \eqref{eq:safeset}.

For Stage II, $D(a_t,a_{t-1}^{\mathrm{exec}})$ in $\mathbf{P2}$ is computed by the squared Euclidean distance between two consecutive fast actions, given by $D(a_t,a_{t-1}^{\mathrm{exec}})=\sum_{u,k}\big(x_{u,k}(t)-x_{u,k}(t\!-\!1)\big)^2+\sum_{u,k}\big((p_{u,k}(t)-p_{u,k}(t\!-\!1))/P^{\mathrm{rb}}\big)^2$, where the power differences are normalized by $P^{\mathrm{rb}}$.
The baseline action $a_t^{\mathrm{base}}$ retains the configuration of $a_{t-1}^{\mathrm{exec}}$ and obtains $(x_{u,k}(t),p_{u,k}(t))$ from one convex solve that minimizes $\sum_{u,k}p_{u,k}(t)$ subject to $c_1$--$c_4$ and $e_1$--$e_3$.
Since only one linear objective is involved without any slacks, it is cheap and fast to solve at each epoch.

\subsection{Roles of Four xApp Agents}\label{sec:uc-agents}
Each agent acts as an xApp running in the Near-RT RIC, empowered by either an external LLM in real-time or a scripted rule, and in this case study, a total of $4$ agents (i.e., $N=4$) are instantiated targeting different tasks.
The first agent, namely \emph{QoS assurance agent}, targets improving the rate KPI $R_u(a_t,s_t)$ in \eqref{eq:rate} only for the protected users in $\mathcal{U}^{\mathrm{prot}}$, which is set as a hard target with a priority class of $1$ (i.e., $\ell=1$).
As for each non-protected user, the same QoS agent has a soft rate target by controlling $R_u(a_t,s_t)$ to reduce the backlog $Q_u(t)$ of \eqref{eq:queue}.
Apart from this, an \emph{energy saving} agent focuses on a soft target of reducing power consumption $E_b(a_t)$ in \eqref{eq:energy} for each cell, while an \emph{interference coordination} agent exists aiming at a soft target of caused interference $I_b^{\mathrm{out}}(a_t,s_t)$ in \eqref{eq:leak}.
Finally, we assume that there is a \emph{load balancing} agent posting per-cell caps on a soft target of load $\rho_b(a_t)$ in \eqref{eq:load}.
Note that none of these four agents directly commands any RAN-control action variables $\big(a_t^{\mathrm{cfg}},\,a_t^{\mathrm{rt}}\big)$ in \eqref{eq:actinst} and \eqref{eq:actinstrt}, and they only affect the action results via the generated proposals.

\subsection{xTRUCE Arbitration and Exact Prices}\label{sec:uc-results}
By substituting \eqref{eq:rate}--\eqref{eq:load} into $g_{i,m}(a_t,s_t)$ of \eqref{eq:req}, $\mathbf{P1}$ and $\mathbf{P2}$ reduce to programs over the fast variables $(x_{u,k}(t),p_{u,k}(t))$ and the slacks $\bm{\xi}$, with $a_t^{\mathrm{cfg}}$ fixed between configuration updates.
From Section~\ref{sec:uc-agents}, the hard targets are only the protected-rate instances at $\ell=1$, thereby $\mathbf{P1}$ reduces to solving $V_1(\bm{\xi})$ only.
Building further upon \eqref{eq:classviol}, $\mathbf{P1}$ becomes
\begin{align}
v_1^{\star}=\min_{a_t^{\mathrm{rt}},\ \bm{\xi}\geqslant 0}\quad
&\sum_{u\in\mathcal{U}^{\mathrm{prot}}}\beta_{u}\,\xi_{u}^{2}\label{eq:p1inst}\\
{\rm s.t.}\quad
&\theta_{u}-R_u(a_t,s_t)\leqslant\xi_{u},\ \forall u\in\mathcal{U}^{\mathrm{prot}},\tag{\ref{eq:p1inst}a}\\
&a_t=\big(a_t^{\mathrm{cfg}},\,a_t^{\mathrm{rt}}\big)\in\Hset_t,\tag{\ref{eq:p1inst}b}
\end{align}
where $(\theta_u,\beta_u,\xi_u)$ abbreviates the requested rate target, weight, and slack of the corresponding hard pair $(i,m)\in\Rset_t$.

As for the instantiation of $\mathbf{P2}$, the soft rate, energy, interference, and load targets of the four agents enter \eqref{P2} as the terms $[g_{i,m}(a_t,s_t)]_{+}^{2}$ by the same substitution, thus giving:
\begin{align}
\min_{a_t^{\mathrm{rt}},\ \bm{\xi}\geqslant 0}\ \
&\sum_{u\in\mathcal{U}\setminus\mathcal{U}^{\mathrm{prot}}}\!\!\!\beta_{u}\big[\theta_{u}-R_u(a_t,s_t)\big]_{+}^{2}
+\eta\,D(a_t,a_{t-1}^{\mathrm{exec}})\notag\\
&+\sum_{b\in\mathcal{B}}\Big(\!\beta_{b}^{\mathrm{E}}\big[E_b(a_t)-\theta_{b}^{\mathrm{E}}\big]_{+}^{2}
\!+\!\beta_{b}^{\mathrm{I}}\big[I_b^{\mathrm{out}}(a_t,s_t)-\theta_{b}^{\mathrm{I}}\big]_{+}^{2}\notag\\
&\qquad\ \ +\beta_{b}^{\rho}\big[\rho_b(a_t)-\theta_{b}^{\rho}\big]_{+}^{2}\Big)\label{eq:p2inst}\\
{\rm s.t.}\ \
&\text{(\ref{eq:p1inst}a), (\ref{eq:p1inst}b)},\ \
\sum_{u\in\mathcal{U}^{\mathrm{prot}}}\!\!\beta_{u}\,\xi_{u}^{2}\leqslant v_1^{\star}+\epsilon_1.\tag{\ref{eq:p2inst}a}
\end{align}
Herein, $\big(\theta_b^{\mathrm{E}},\theta_b^{\mathrm{I}},\theta_b^{\rho}\big)$ and $\big(\beta_b^{\mathrm{E}},\beta_b^{\mathrm{I}},\beta_b^{\rho}\big)$ are the per-cell soft caps and their weights posted by the energy saving, interference coordination, and load balancing agents on \eqref{eq:energy}, \eqref{eq:leak}, and \eqref{eq:load}, respectively.
Constraint (\ref{eq:p2inst}a) corresponds to constraint (\ref{P2}b) under this instance, using the single output $v_1^{\star}$ of \eqref{eq:p1inst}.

To address these two instantiated problems, we give the following theorem to establish the solution rationale:
\begin{theorem}[Guarantee on convexity and exact prices]\label{th:convex}
For any given $s_t$ and $a_t^{\mathrm{cfg}}$, both \eqref{eq:p1inst} and \eqref{eq:p2inst} are convex w.r.t. $\big(a_t^{\mathrm{rt}},\bm{\xi}\big)$.
If $\epsilon_1>0$ and there exists an action $\bar a_t=(a_t^{\mathrm{cfg}},\bar a_t^{\mathrm{rt}})\in\Hset_t$, such that $R_u(\bar a_t,s_t)>R_u^{\min}$,
$\forall u\in\mathcal U^{\mathrm{prot}}$, then both problems hold strong duality with the dual optima attained, and $\bm{\mu}_t$ in \eqref{eq:cert} are the exact Lagrange multipliers of \eqref{eq:p2inst}.
\end{theorem}
\begin{IEEEproof}
Please see Appendix~B.
\end{IEEEproof}

Theorem~\ref{th:convex} ensures that \eqref{eq:p1inst} and \eqref{eq:p2inst} can be directly handed to some off-the-shelf convex programming tools (such as CVXPY~\cite{diamond2016cvxpy}) to finalize $v_1^{\star}$ and $a_t^{\star}$ in polynomial time.
Especially, at every $T^{\mathrm{cfg}}$-th epoch, xTRUCE solves \eqref{eq:p1inst} under each candidate $a_t^{\mathrm{cfg}}$, records every completed evaluation in $\mathcal F_t$, and solves \eqref{eq:p2inst} once under the lexicographically best candidate if all candidates finish by $\tau$.
As such, xTRUCE then assembles $\Zcert_t$ as the feedback to xApps via \eqref{eq:cert} and outputs a safe action $a_t^{\mathrm{exec}}$ via \eqref{eq:anytime}.
It should be emphasized that Theorem~\ref{th:convex} only characterizes the properties of this specific use case rather than of general O-RAN scenarios.
Even in non-convex or delay-exceeded cases, the actions would still stay on the verified outcome of \eqref{eq:anytime}, thereby only the intent-layer optimality is lost, yet without compromising any gNB control safety.

\section{Experimental Evaluation}\label{sec:eval}
In this section, xTRUCE is evaluated exactly in line with the four design goals outlined in Section \ref{sec:arch-loops}: i) xApp proposal-independent safety, ii) priority-consistent target satisfaction, iii) certificate-guided LLM renegotiation, iv) delay-aware safe execution.
To this end, we conduct multi-cell Python-based simulations and OTA experiments on an OAI/FlexRIC-based O-RAN testbed, using the same xTRUCE implementation.

\subsection{Simulation and OTA Setup}
In the simulations, we instantiate the full control chain of Fig.~\ref{fig:arch}, from the policy output of the SMO/Non-RT RIC, through the xApps and xTRUCE at the Near-RT RIC, to a multi-cell RAN.
A configured operator policy represents the A1 policy of the intent-orchestration rApp, which by default enforces the protected-user rate floors as the rigid limit $e_1$ and ranks the protected-rate targets at $\ell=1$, the non-protected-user rate, energy, and interference targets at $\ell=2$, and the load caps at $\ell=3$.
The same four xApp agents of Section~\ref{sec:uc-agents} use deterministic rules in the safety, priority, and scalability sweeps, while the certificate-renegotiation simulations connect live ChatGPT-5.6 Terra to the QoS xApp and Claude-Sonnet-5 to the energy xApp through their respective vendor APIs.
All xApp realizations reuse one structured schema over ZeroMQ~\cite{hintjens2013zeromq}, which carries the RAN observations, JSON proposals, and their own certificate.
In addition, each accepted proposal stays valid for two epochs unless renewed, decoupling LLM generation from per-epoch arbitration.

Afterward, we implement Algorithm~\ref{alg:arbiter} in Python and formulate its two arbitration stages as parameterized exponential-cone programs in CVXPY, solved at each epoch through a warm-started fallback chain of CLARABEL, ECOS, and SCS.
All simulation trials run on a single $2.6$-GHz six-core Intel i7 CPU, and each setting uses $20$ screened seeds whose network realizations keep the protected-user rate floors physically reachable.
Below xTRUCE, a pair of replaceable state and actuation adapters connects the arbitration core to a custom epoch-level gNB simulator that implements the channel, queue, and power models of Section~\ref{sec:usecase}, and rebinding this pair to the live testbed reuses the whole upper chain in the OTA experiments.
For practical consistency, the simulations do not differentiate the fast actions of \eqref{eq:actinstrt} over individual RBs, so that xTRUCE arbitrates one overall spectrum share and one total transmit power per user, while each committed floor is planned against the measured interference plus the maximum one-epoch interference increase allowed by $e_2$.
All gNB-related parameters of the simulations and OTA experiments are summarized in Table~\ref{tab:params} unless otherwise stated.

\begin{table}[t]
\caption{Default Parameters in Simulation and OTA Evaluation}
\label{tab:params}
\centering
\footnotesize
\setlength{\tabcolsep}{2.8pt}
\renewcommand\arraystretch{1.3}
\begin{tabular}{|>{\raggedright\arraybackslash}p{0.34\columnwidth}|>{\raggedright\arraybackslash}p{0.27\columnwidth}|>{\raggedright\arraybackslash}p{0.29\columnwidth}|}
\hline
\textbf{Parameters} & \textbf{Simulation} & \textbf{OTA Testbed}\\
\hline
Number of cells / users ($B$ / $U$) & $4$ / $20$ (hexagonal, $500$-m ISD) & $1$ / $4$ (single indoor cell, 5G modules)\\
\hline
Spectrum ($K\times W$) & $12\times360$~kHz & $24\times360$~kHz\\
\hline
Protected users ($|\mathcal{U}^{\mathrm{prot}}|$) / Rate floor ($R_u^{\min}$) & $3$ protected users / $2$~Mbps & $1$ protected user / $0.5$~Mbps\\
\hline
Change limits ($\Delta^{x}$, $\Delta^{p}$, $\Delta^{\mathrm{act}}$, $\Delta^{\mathrm{str}}$) & $0.25$, $0.25$~W, $1$ cell, $3$ users & $0.25$, $0.25$~W, N/A, N/A\\
\hline
Channel model & UMa path loss~\cite{3gpp36814} + $8$-dB shadowing; AR(1) Rayleigh ($\rho=0.349$)~\cite{baddour2005ar} & Indoor n48 OTA link (SNR $\approx21$~dB); per-RB rate calibrated from measurement\\
\hline
Traffic per user ($\mathbb E[\lambda_u(t)]/\tau$) & $6$~Mbps (exponential arrivals) & $3.2$~Mbps (payload offered)\\
\hline
Epoch length ($\tau$) & $1$~s & $0.1$~s\\
\hline
Powers ($P_b^{\max}$/$P_b^{\mathrm{cir}}$/$P^{\mathrm{rb}}$) & \multicolumn{2}{c|}{$10$/$50$/$2$~W}\\
\hline
Noise power ($\sigma^2$) & \multicolumn{2}{c|}{$1.15\times10^{-14}$~W}\\
\hline
Optimality tolerance ($\epsilon_\ell$) & \multicolumn{2}{c|}{$10^{-4}(1+v_\ell^{\star})$}\\
\hline
Target weight ($\beta_{i,m}$) & \multicolumn{2}{c|}{$1$ for all $(i,m)\in\Rset_t$}\\
\hline
\end{tabular}
\end{table}
As for the OTA experiments, we replace the simulated RAN environment with a live OAI/FlexRIC adapter while keeping the xTRUCE arbitration core unchanged, where deterministic realizations of the same four xApp roles are used for repeatable comparisons.
As shown on the right side of Fig.~\ref{fig:arch}, xTRUCE runs as a standalone Dockerized arbitration service co-located with the FlexRIC Near-RT RIC~\cite{schmidt2021flexric} on an NVIDIA DGX Spark~\cite{nvidia2026dgxspark}.
On top of the Python SDK of FlexRIC, the live adapter collects the RAN-state input $s_t$ from the E2 MAC/RLC/PDCP telemetry and maps the spectrum-share component of $a_t^{\mathrm{exec}}$ into integer quotas over $24$ RBs, which are submitted to the OAI gNB~\cite{kaltenberger2019oai5gnr} through the FlexRIC Slice SM.
Herein, we measure the rate delivered by one RB once in advance, enabling xTRUCE to relate each rate target to an RB quota.
Note that the OTA experiments only test the RB allocation rather than the power allocation of xTRUCE, since the deployed FlexRIC Slice SM path carries no power command~\cite{schmidt2021flexric}.
The same host runs the OAI 5G Core (v2.1.10), while an NI USRP-2901 SDR~\cite{ni2017usrp2901} operates the OTA cell on band-n48 at $3.61$~GHz with a $30$-kHz subcarrier spacing.
Meanwhile, four commercial 5G user equipment modules simultaneously attach over the air, each carrying a continuous $3.2$-Mbps downlink payload offered from the core-network side.
Accordingly, the OTA experiments focus on evaluating the live E2 measurement-and-control loop, whether every RB-allocation action delivered to the gNB satisfies the applicable constraints under conflicting xApp proposals, and whether \eqref{eq:anytime} rejects late or unverified arbitration results.

For comparison purposes, both environments replace only xTRUCE with two benchmark schemes of gNB action execution under conflicting xApp proposals: 1) a \emph{Direct} scheme, which directly turns the most demanding target of each xApp into the corresponding gNB action without checking its joint feasibility; and 2) a \emph{Clipping} scheme, which further clips the resulting actions from the Direct scheme and proportionally rescales them to satisfy physical constraints $c_1$--$c_4$ in \eqref{eq:phys} at the same time.
Neither benchmark acts on the interference and load xApp targets,\footnote{This is because the QoS and energy xApps have a relatively direct target-to-action mapping relationship, whereas the targets of interference and load xApps require joint coordination for gNB action determination.} operates on the operator priority and limits, or returns conflict certificates, hence their gaps to xTRUCE precisely quantify what the two-stage arbitration contributes.

\subsection{Simulation Results}\label{sec:ev-sim}
\begin{figure}[t]
\centering
\includegraphics[width=\columnwidth]{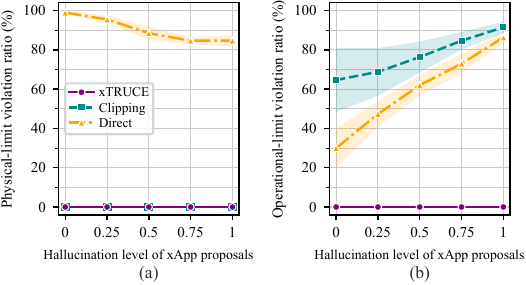}
\caption{Rigid-constraint violation ratios vs. varying proposal hallucination levels: (a) ratio of epochs violating $c_1$--$c_4$; (b) ratio of epochs violating $e_1$. The shaded bands span the $95\%$ confidence intervals over $20$ screened seeds.}
\label{fig:e1}
\end{figure}
\emph{(i) Proposal-independent gNB control safety:} Fig.~\ref{fig:e1} evaluates xTRUCE regarding the performance of rigid-constraint violation as the hallucination level of xApp proposals increases from $0$ (no hallucination) to $1$ (complete hallucination).
Here, we set the hallucination level following~\cite{huang2025resilience} to jointly control the probability of corrupting each xApp proposal and its deviation from the valid target.
Fig.~\ref{fig:e1}(a) first presents the violation ratios for physical limits $c_1$--$c_4$ over $600$ epochs at each hallucination level under the three schemes.
It is observed that the gNB control actions output by xTRUCE and Clipping never violate any physical limit, whereas Direct violates at least one physical limit in $99\%$ of the epochs even without injected hallucination and in $85$ to $96\%$ once hallucinated.
Under the same settings, Fig.~\ref{fig:e1}(b) then measures the violation ratios of the service-level operational limit $e_1$ (i.e., a committed rate floor of $2$ Mbps for each protected user in~\eqref{eq:ramp}).
It shows that xTRUCE still keeps zero violation, and Direct violates $e_1$ in $30$ to $86\%$ of the epochs.
Clipping performs even worse here, leaving at least one user below the rate floor in $65$ to $92\%$ of the epochs, which means its action stays physically legal as in Fig.~\ref{fig:e1}(a), but its committed services are lost.
Therefore, xTRUCE is the only scheme that preserves both physical feasibility and operator-committed rate floors across all tested proposal hallucination levels, thereby demonstrating its proposal-independent gNB control safety.

\begin{figure}[t]
\centering
\includegraphics[width=\columnwidth]{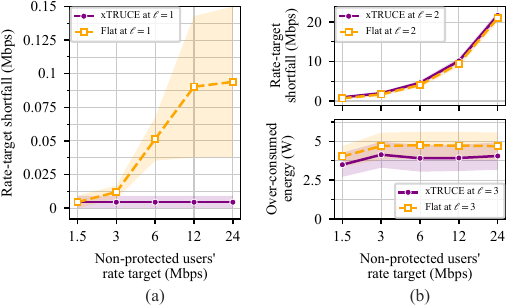}
\caption{Realized KPI-target shortfall comparisons under resource overload, induced by raising the non-protected-user rate target: (a) protected-user rate-target shortfalls ($\ell=1$); (b) non-protected-user rate-target shortfalls ($\ell=2$) (top), and over-consumed energy beyond caps ($\ell=3$) (bottom). Curves and shaded bands show the mean and $95\%$ confidence intervals, respectively.}
\label{fig:e2}
\end{figure}
\emph{(ii) Priority-consistent KPI satisfaction:} Next, Fig.~\ref{fig:e2} evaluates xTRUCE's capability of priority-consistent KPI target satisfaction as resource contention intensifies across three operator priority classes (i.e., $\ell=1,2,3$).
In this test, we declare three hard targets, i.e., offering a $3$-Mbps protected-user rate target at $\ell=1$, a non-protected-user rate target at $\ell=2$, and a $53$-W per-cell energy cap at $\ell=3$, whereas the resource contention sweep is controlled by raising the $\ell=2$ rate target from $1.5$ to $24$ Mbps.
Moreover, a dedicated benchmark scheme, namely \textit{Flat}, is devised that uses the same arbitration mechanism as xTRUCE but assigns all targets the same priority.
Fig.~\ref{fig:e2}(a) first depicts the shortfall performance of $\ell=1$ under the two schemes.
It is seen that the gNB control action derived from xTRUCE always keeps the highest-priority rate-target shortfall near zero regardless of the resource overload intensity, whereas the shortfall of Flat grows with the resource contention, up to about $0.1$~Mbps.
Fig.~\ref{fig:e2}(b) then tests the shortfall performance at $\ell =2$ (top) and $\ell=3$ (bottom) under the same overload levels.
With xTRUCE, the rate-target shortfall at non-protected users grows to $21.7$~Mbps and the over-consumed energy reaches $4.1$~W at the most intensive resource contention, which are precisely the priority concession results described in Theorem~\ref{th:priority}.
In contrast, the Flat curves show that disregarding KPI priorities buys little in return, for instance, the shortfall at $\ell =2$ shrinks by no more than $0.7$~Mbps and the over-consumed energy at $\ell =3$ even exceeds our xTRUCE by up to $0.8$~W.
The above results prove that, however far the resource demand scales, xTRUCE ensures the protection of higher-priority targets before considering lower-priority ones, thereby showcasing its ability to meet KPIs in accordance with operator-assigned priorities.

\begin{figure}[t]
\centering
\includegraphics[width=\columnwidth]{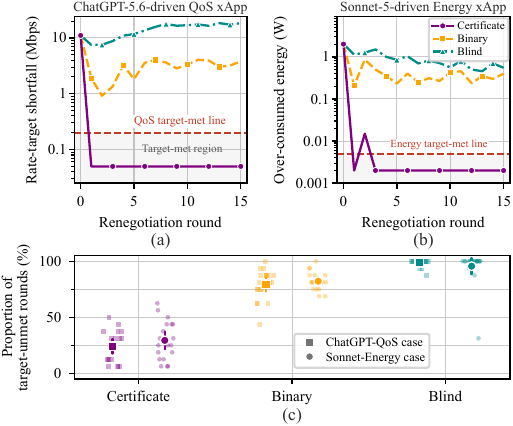}
\caption{Certificate-guided renegotiation performance of two live LLMs under three conflict feedback modes: (a) rate-target shortfall w.r.t. ChatGPT-5.6-driven QoS xApp; (b) over-consumed energy w.r.t. Sonnet-5-driven energy xApp; (c) proportion of target-unmet rounds with per-run points and means.}
\label{fig:e3}
\end{figure}
\emph{(iii) Certificate-guided renegotiation with live LLM xApps:} Further, Fig.~\ref{fig:e3} validates the conflict renegotiation effectiveness guided by xTRUCE's certificates, where we study two live LLM-xApp cases, i.e., a ChatGPT-5.6 Terra-driven QoS xApp and a Claude-Sonnet-5-driven energy xApp, each aiming to revise an initially infeasible KPI intent.
For comparison, two feedback benchmarks are employed: a \emph{Binary} mode returning only a one-bit indication of whether the KPI target is met and a \emph{Blind} mode returning no arbitration feedback to xApps.
Together with xTRUCE's \emph{Certificate} mode, all three feedback modes observe the same E2-telemetry fields and share the same arbitration process, testing over a total of $16$ renegotiation rounds (taking the median of $20$ runs for each round).
Here, a round is regarded as target-met if its performance gap does not exceed $0.2$~Mbps in the ChatGPT-QoS case or $0.005$~W in the Sonnet-Energy case, i.e., the red lines in the figures.
With an initial $20$-Mbps rate target while keeping a $5.5$-Mbps rate floor, Fig.~\ref{fig:e3}(a) first depicts the renegotiation effectiveness of the three modes in the ChatGPT-QoS case.
It is observed that Certificate guides the ChatGPT-5.6 agent to adjust its proposal and successfully meet the rate target using only one round, clearly outperforming Binary and Blind.
Then, with an initial $50$-W energy cap, Fig.~\ref{fig:e3}(b) examines the Sonnet-Energy case.
Similarly, Certificate reaches the energy-consumption line within three rounds, while Binary and Blind consistently fail to effectively guide LLM agents.
Note that the round-$2$ rebound is because Sonnet-5 adjusts the energy cap to exactly the energy consumption at round-$1$, so the naturally fluctuating power consumption briefly exceeds the line.
Moreover, Fig.~\ref{fig:e3}(c) counts the proportion of target-unmet rounds, where Certificate yields the minimum in both cases with a mean of $24\%$ against $79$ and $99\%$ (ChatGPT-QoS) and $29\%$ against $82$ and $96\%$ (Sonnet-Energy), since only Certificate reveals to LLMs where the feasibility boundary lies and how it should be relaxed across its requests.
The above demonstrates that xTRUCE can replace these vendor-dependent trial-and-error methods with explicit per-target certificates, enabling rapid, low-overhead recovery from infeasible xApp proposals and ensuring consistency across differing LLM intents.

\begin{figure}[t]
\centering
\includegraphics[width=\columnwidth]{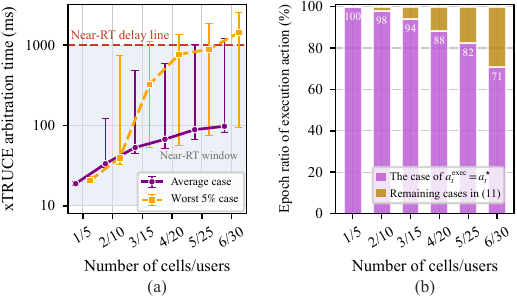}
\caption{xTRUCE arbitration timing and action selection under varying physical RAN capacities: (a) medians and min--max ranges over $20$ seeds for two timing measures; (b) ratio of epochs executing different execution action cases.}
\label{fig:e4}
\end{figure}

\emph{(iv) Delay-safe action execution:} In Fig.~\ref{fig:e4}, we examine the arbitration time of xTRUCE and the selected case of \eqref{eq:anytime} as the RAN capacity grows from one to six cells with five users each.
For accurate measurement, we exclude network realizations in which the protected-user rate floors are physically unreachable, and the retained runs cover $3000$ epochs in total.
Specifically, Fig.~\ref{fig:e4}(a) first presents the per-epoch xTRUCE arbitration time at each network size, where the average case is derived from taking the median of all per-epoch arbitration times, while the worst $5\%$ case means the top $5\%$ of epochs taking the longest arbitration times.
It is seen that the average-case time stays between $19$ and $98$ ms, far below $1$s at each median point, and the worst $5\%$ case begins to cross the Near-RT line at $30$ users, which is just because the arbitration complexity increases with the network scale.
Fig.~\ref{fig:e4}(b) then counts the ratio of epochs in the case of $a_t^{\mathrm{exec}}=a_t^{\star}$ (i.e., a verified Stage-II action is available within $\tau$) against that in the remaining cases of \eqref{eq:anytime} (i.e., Stage I\&II are not completed in time).
The results show that the epoch ratio executing $a_t^{\star}$ is first $100\%$ at $5$ users, and then gradually and slowly declines from $98\%$ at $10$ users to $71\%$ at $30$ users.
Such a downtrend is apparent since the arbitration time grows with the RAN capacity as observed in Fig.~\ref{fig:e4}(a), making more epochs hit the delay limit and thus fall back to the remaining cases of \eqref{eq:anytime}.
Execution records further certify that all $3000$ epochs use either $a_t^{\star}$ or one of the remaining verified cases in \eqref{eq:anytime}, which means unavailable or delay-exceeded results are never used as execution actions across different network sizes.

\begin{figure*}[!t]
\centering
\includegraphics[width=0.98\textwidth]{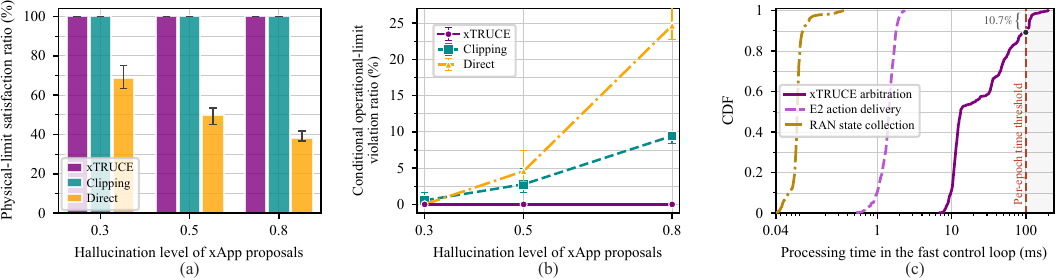}
\caption{OTA evaluation of xTRUCE under four hallucinated xApps: (a) ratios of epochs meeting $c_1$--$c_4$; (b) ratios of epochs violating $e_1$ among those actions meeting $c_1$--$c_4$; (c) empirical CDFs of fast-control-loop processing times across $540$ control epochs. Error bars in (a) and (b) span three repeated OTA tests.}
\label{fig:ota}
\end{figure*}

\subsection{OTA Testbed Results}\label{sec:ev-ota}
Beyond the above simulations, we finally validate xTRUCE over the air on the O-RAN-compliant OAI/FlexRIC testbed, as illustrated in Fig.~\ref{fig:ota}, where the same four xApps inject hallucinated proposals into a live E2 control loop.

\emph{(v) Proposal-independent gNB control safety over the air:} We first re-examine the real-world gNB control safety under xTRUCE at three fixed proposal hallucination levels of $0.3$, $0.5$, and $0.8$, and compare the results with the Clipping and Direct benchmarks.
Similar to Fig.~\ref{fig:e1}, the ratio of epochs in which the control actions satisfy the physical limits $c_1$--$c_4$ and the ratio of epochs in which these physically feasible actions still violate the operational limit $e_1$ are shown in Fig.~\ref{fig:ota}(a) and Fig.~\ref{fig:ota}(b), respectively.
It is seen that xTRUCE remains at a $100\%$ satisfaction ratio for both checks at each hallucination level, whereas Direct satisfies the physical limits in only $69$ down to $38\%$ of the epochs.
Particularly, Fig.~\ref{fig:ota}(b) demonstrates that even physically feasible actions can still violate $e_1$ as the hallucination grows, for instance, up to $9.4\%$ for Clipping and $24.7\%$ for Direct at the highest level, whereas xTRUCE stays at zero.
The above OTA results reconfirm that physical clipping alone cannot preserve the operator-defined service limit, and only xTRUCE preserves gNB control safety.

\emph{(vi) Delay-safe execution in the live E2 loop:} Fig.~\ref{fig:ota}(c) then tests the cumulative distribution function (CDF) of the per-epoch processing time of the whole live E2 loop.
Concretely, we break the aforementioned fast E2 control loop into three consecutive stages, i.e., RAN state collection over the E2 interface, xTRUCE arbitration, and gNB control action delivery over E2.
It is observed that RAN state collection and E2 action delivery complete within $0.4$ and $2.4$~ms, respectively, in all $540$ epochs, while arbitration dominates with a median of $13.3$~ms.
Moreover, the arbitration CDF meets the per-epoch time threshold, i.e., the $0.1$-s epoch length $\tau$, exactly at $89.3\%$ (the marked dot), leaving $10.7\%$ of xTRUCE calls beyond the threshold.
Our execution logs also show that all $58$ late calls and all $33$ in-time calls without a verified $a_t^{\star}$ still use the remaining verified cases in \eqref{eq:anytime}.
Hence, consistent with the simulation results in Fig.~\ref{fig:e4}, xTRUCE still preserves its safety checks and late-result rejection on a real O-RAN stack.

\section{Conclusions}\label{sec:concl}
In this paper, we proposed the xTRUCE system, a provable realization for a multi-xApp conflict-mitigation function in the emerging agentic O-RAN.
It confined untrusted xApp agents to structured proposals, arbitrated these proposals in two stages along the operator priorities in a provably safe way, and returned a verified gNB control action via E2 and an exact conflict certificate for LLM renegotiation.
We finally implemented its multi-process prototype and evaluated it through both simulations and OTA experiments.
Results demonstrated a variety of guarantees, covering safe gNB control under severely hallucinated proposals, priority-consistent KPI target satisfaction under resource overload, certificate-guided recovery of infeasible LLM intents across model vendors, and delay-safe per-epoch execution.
This work can serve as a pioneer in exploring safe multi-xApp control for agentic O-RAN, where the Near-RT RIC enforces physical limits, operator priorities, and timely execution without relying on the correctness of individual xApps.
Since xTRUCE addresses only multi-xApp conflicts and treats interference measured in the previous epoch as constant, extending it to LLM-driven multi-rApps at the Non-RT RIC and incorporating dynamic interference drift would be two promising directions for future work.

\appendices
\section{Proof of Theorem~\ref{th:priority}}
We fix epoch $t$ throughout this proof.
For any $a_t\in\Hset_t$, the smallest feasible slack is $\xi_{i,m}=[g_{i,m}(a_t,s_t)]_+$, hence $\mathbf{P1}$ can be analyzed through $V_\ell(a_t)$.
Let $\Hset^{(0)}=\Hset_t$, $v_\ell^{\star}=\min_{a_t\in\Hset^{(\ell-1)}}V_\ell(a_t)$, and $\Hset^{(\ell)}=\{a_t\in\Hset^{(\ell-1)}\mid V_\ell(a_t)=v_\ell^{\star}\}$ for every $\ell\in\Lset$.
The stated compactness of $\Hset_t$ ensures that every minimum is attained and every $\Hset^{(\ell)}$ is nonempty.
If any $a_t\in\Hset_t$ yields a lexicographically smaller vector, then at the first differing class $\ell_0$, it would belong to $\Hset^{(\ell_0-1)}$ while satisfying $V_{\ell_0}(a_t)<v_{\ell_0}^{\star}$.
This contradicts the definition of $v_{\ell_0}^{\star}$, thereby the Stage-I vector is uniquely lexicographically minimal.
In addition, any action in $\Hset^{(L)}$ with its smallest feasible slacks is feasible for $\mathbf{P2}$ for arbitrary $\epsilon_\ell\geqslant0$.
Every feasible pair $(a_t,\bm{\xi})$ of $\mathbf{P2}$ satisfies $V_\ell(a_t)\leqslant V_\ell(\bm{\xi})\leqslant v_\ell^{\star}+\epsilon_\ell$, so the same bound holds for $a_t^{\star}$.
If all $\epsilon_\ell=0$ and $(V_1(a_t^{\star}),\cdots,V_L(a_t^{\star}))$ differs from $(v_1^{\star},\cdots,v_L^{\star})$, its first differing component would be smaller, contradicting the lexicographic minimality proved above.
Hence, $V_\ell(a_t^{\star})=v_\ell^{\star}$ for every $\ell$ when all $\epsilon_\ell=0$.
This completes the proof.

\section{Proof of Theorem~\ref{th:convex}}
\emph{Convexity:}
First, given fixed $s_t$, the simplified function $f(p)=W\log_2\big(1+G_{u,k}(t)\,p/(\sigma^2+I_{u,k}(t))\big)$ is obviously concave in $p\geqslant 0$.
Then, each summand of \eqref{eq:rate} equals the perspective $x f(p/x)$, which is jointly concave in $(x,p)$ for $x>0$ and extends concavely by $0$ at $x=0$, hence the rate $R_u(a_t,s_t)$ is concave.
In this case, every rate-target shortfall function is convex, while $\big[Q_u(t)-\tau R_u(a_t,s_t)\big]_{+}+\lambda_u(t)$, i.e., the queue value in \eqref{eq:queue}, is also convex.
Meanwhile, the KPIs \eqref{eq:energy}, \eqref{eq:leak}, and \eqref{eq:load} are clearly linear w.r.t. $a_t^{\mathrm{rt}}$, as are all limits in \eqref{eq:phys} and $e_2$--$e_3$ in \eqref{eq:ramp}, while the rate limit $R_u\geqslant R_u^{\min}$ defines a convex set by the concavity of $R_u$.
Based on the above, the objective of \eqref{eq:p1inst} is a convex quadratic w.r.t. $\bm{\xi}$, whereas the objective of \eqref{eq:p2inst} is a weighted sum of squared positive-part terms of convex functions and the convex quadratic action-change term $D$.
Therefore, both objectives are convex and every constraint defines a convex set.

\emph{Strong duality and exact prices:}
After expressing $e_2$--$e_3$ as affine inequality pairs, meeting Slater's condition requires strictness only for $e_1$, (\ref{eq:p1inst}a), and the last constraint in (\ref{eq:p2inst}a).
Choosing $\bar\xi_u>[\theta_u-R_u(\bar a_t,s_t)]_+$ establishes refined Slater for \eqref{eq:p1inst}.
For \eqref{eq:p2inst}, let $a_t^{(1)}$ be a Stage-I minimizer with its minimal slacks, and define $a_t(\rho)=(a_t^{\mathrm{cfg}},(1-\rho)a_t^{(1),\mathrm{rt}}+\rho\bar a_t^{\mathrm{rt}})$.
Convexity of the fixed-configuration feasible set and concavity of $R_u$ make $a_t(\rho)$ feasible with $e_1$ strict.
As $\epsilon_1>0$, sufficiently small $(\rho,\delta)>0$ with $\xi_u=[\theta_u-R_u(a_t(\rho),s_t)]_++\delta$ make all remaining non-affine inequalities strict.
Thus, refined Slater holds~\cite{boyd2004convex}, proving strong duality, dual attainment, and the stated multiplier interpretation, which ends the proof.

\bibliographystyle{IEEEtran}
\bibliography{refs}

\end{document}